\documentclass{llncs}

\usepackage{xspace}
\usepackage{array}
\usepackage{subfigure}
\usepackage{subfig}
\usepackage{xcolor}
\usepackage{epsfig}
\usepackage{gensymb}
\usepackage[english]{babel}
\usepackage{graphicx}
\usepackage{amssymb}
\usepackage{multirow}
\usepackage{todonotes}
\usepackage{enumerate}
\usepackage{paralist}
\usepackage{amsmath}
\usepackage{cases}
\usepackage{siunitx}
\usepackage{fullpage}
\graphicspath{{img/}}

\newcommand{\remove}[1]{}

\renewcommand{\int}{int}

\renewenvironment{proof}
{{\bf Proof:}}{\hspace*{\fill}$\Box$\par\vspace{2mm}}

\begin{document}
\title{Drawing Graphs as Spanners\thanks{This project has received funding from the European Union's Horizon 2020 research and innovation programme under the Marie Sk{\l}odowska-Curie grant agreement No 734922, the Natural Sciences and Engineering Research Council of Canada, and by MIUR Projects ``MODE'' under PRIN 20157EFM5C and ``AHeAD'' under PRIN 20174LF3T8.}}
\author{Oswin Aichholzer$^1$ \and Manuel Borrazzo$^2$ \and Prosenjit Bose$^3$ \and Jean Cardinal$^4$ \and \\Fabrizio Frati$^2$ \and Pat Morin$^3$ \and Birgit Vogtenhuber$^1$}
\institute{$^1$Institute for Software Technology,  Graz University of Technology, Graz, Austria\\\email{\{bvogt,oaich\}@ist.tugraz.at}\\
$^2$Roma Tre University, Rome, Italy \\\email{\{manuel.borrazzo,fabrizio.frati\}@uniroma3.it}\\
$^3$School of Computer Science, Carleton University, Ottawa, Canada\\\email{\{jit,morin\}@scs.carleton.ca}\\
$^4$Computer Science Department, Universit\'e Libre de Bruxelles (ULB), Brussels, Belgium\\\email{jcardin@ulb.ac.be}}
\authorrunning{O. Aichholzer \etal}
\maketitle

\begin{abstract}
	We study the problem of embedding graphs in the plane as good geometric spanners. That is, for a graph $G$, the goal is to construct a straight-line drawing $\Gamma$ of $G$ in the plane such that, for any two vertices $u$ and $v$ of $G$, the ratio between the minimum length of any path from $u$ to $v$ and the Euclidean distance between $u$ and $v$ is small. The maximum such ratio, over all pairs of vertices of~$G$, is the \emph{spanning ratio} of $\Gamma$.
	
	First, we show that deciding whether a graph admits a straight-line drawing with spanning ratio~$1$, a proper straight-line drawing with spanning ratio~$1$, and a planar straight-line drawing with spanning ratio~$1$ are NP-complete, $\exists \mathbb R$-complete, and linear-time solvable problems, respectively, where a drawing is proper if no two vertices overlap and no edge overlaps a vertex.
	
	Second, we show that moving from spanning ratio $1$ to spanning ratio $1+\epsilon$ allows us to draw every graph. Namely, we prove that, for every $\epsilon>0$, every (planar) graph admits a proper (resp.\ planar) straight-line drawing with spanning ratio smaller than~$1+\epsilon$.
	
	Third, our drawings with spanning ratio smaller than~$1+\epsilon$ have large edge-length ratio, that is, the ratio between the length of the longest edge and the length of the shortest edge is exponential. We show that this is sometimes unavoidable. More generally, we identify having bounded toughness as the criterion that distinguishes graphs that admit straight-line drawings with constant spanning ratio and polynomial edge-length ratio from graphs that require exponential edge-length ratio in any straight-line drawing with constant spanning ratio.
\end{abstract}

\section{Introduction}

Let $P$ be a set of points in the plane and let $\mathcal G$ be a geometric graph whose vertex set is~$P$. We say that $\mathcal G$ is a \emph{$t$-spanner} if, for every pair of points $p$ and $q$ in~$P$, there exists a path from $p$ to $q$ in $\mathcal G$ whose total edge length is at most $t$ times the Euclidean distance $\|pq\|$ between $p$ and $q$. The \emph{spanning ratio} of $\mathcal G$ is the smallest real number $t$ such that $\mathcal G$ is a $t$-spanner. The problem of constructing, for a given set $P$ of points in the plane, a sparse (and possibly planar) geometric graph whose vertex set is $P$ and whose spanning ratio is small has received considerable attention; see, e.g.,~\cite{bdlsv-adt-11,bfvv-cr-12,bfvv-pcbds-19,c-tpg-89,dfs-dg-90,dg-lbd-16,m-mdt-04}. We cite here the fact that the Delaunay triangulation of a point set has spanning ratio at least 1.593~\cite{xz-ttb-11} and at most 1.998~\cite{x-sfd-13}, and refer to the survey of Bose and Smid~\cite{bs-pgs-13} for more results. 

In this paper we look at the construction of geometric graphs with small spanning ratio from a different perspective. Namely, the problem we consider is whether it is possible to embed a given abstract graph in the plane as a geometric graph with small spanning ratio. That is, for a given graph, we want to construct a straight-line drawing with small spanning ratio, where the spanning ratio of a straight-line drawing is the maximum ratio, over all pairs of vertices $u$ and $v$, between the total edge length of a shortest path from $u$ to $v$ and $\|uv\|$.

Graph embeddings in which every pair of vertices is connected by a path satisfying certain geometric properties have been the subject of intensive research. As the most notorious example, a \emph{greedy} drawing of a graph~\cite{adf-sgd-12,afg-acgdt-10,ddf-pg-17,eg-sggruhg-11,hz-osgdptt-14,lm-srgems-10,np-egdt-13,pr-crgr-05,rpss-grwli-03,wh-sscgd3pg-14} is such that, for every pair of vertices $u$ and $v$, there is a path from $u$ to $v$ that monotonically decreases the distance to $v$ at every vertex. More restricted than greedy drawings are \emph{self-approaching} and \emph{increasing-chord} drawings~\cite{acglp-sag-12,dfg-icgps-15,npr-osaicd-16}. In a self-approaching drawing, for every pair of vertices $u$ and $v$, there is a \emph{self-approaching} path from $u$ to $v$, i.e., a path $P$ such that $\|ac\|>\|bc\|$, for any three points $a$, $b$, and $c$ in this order along $P$; in an increasing-chord drawing, for every pair of vertices $u$ and $v$, there is a path from $u$ to $v$ which is self-approaching both from $u$ to $v$ and from $v$ to $u$. Even more restricted are \emph{angle-monotone} drawings~\cite{bbcklv-gt-16,dfg-icgps-15,lm-clr-19} in which, for every pair of vertices $u$ and $v$, there is a path from $u$ to $v$ such that the angles of any two edges of the path differ by at most $90^\circ$. Finally, \emph{monotone} drawings~\cite{acdfp-mdg-12,adk-mdt-15,hh-omdt-17,hh-md3pg-15,hr-gstgd-15,kssw-omdt-14} and \emph{strongly-monotone} drawings~\cite{acdfp-mdg-12,fikkms-smdpg-16,kssw-omdt-14} require, for every pair of vertices $u$ and $v$, that a path from $u$ to $v$ exists that is monotone with respect to some direction or with respect to the direction of the straight line through $u$ and $v$, respectively. While greedy, monotone, and strongly-monotone drawings might have unbounded spanning ratio, self-approaching, increasing-chord, and angle-monotone drawings are known to have spanning ratio at most $5.34$~\cite{ikl-sac-99}, at most $2.1$~\cite{r-cic-94}, and at most $1.42$~\cite{bbcklv-gt-16}, respectively. However, not all graphs, and not even all trees~\cite{lm-srgems-10,np-egdt-13}, admit self-approaching, increasing-chord, or angle-monotone drawings.

Our results are the following.
\begin{itemize}
	\item First, we look at straight-line drawings with spanning ratio equal to $1$, which is clearly the smallest attainable value by any graph. We prove that deciding whether a graph admits a straight-line drawing, a proper straight-line drawing (in which no vertex-vertex or vertex-edge overlaps are allowed), and a planar straight-line drawing with spanning ratio $1$ are NP-complete, $\exists \mathbb R$-complete, and linear-time solvable problems, respectively.
	\item Second, we show that allowing each shortest path to have a total edge length slightly larger than the Euclidean distance between its end-vertices makes it possible to draw all graphs. Namely, we prove that, for every $\epsilon>0$, every graph admits a proper straight-line drawing with spanning ratio smaller than $1+\epsilon$ and every planar graph admits a planar straight-line drawing with spanning ratio smaller than $1+\epsilon$. 
	\item Third, we address the issue that our drawings with spanning ratio smaller than $1+\epsilon$ have poor resolution. That is, the \emph{edge-length ratio} of these drawings, i.e., the ratio between the lengths of the longest and of the shortest edge, might be super-polynomial in the number of vertices of the graph. We show that this is sometimes unavoidable, as stars have exponential edge-length ratio in any straight-line drawing with constant spanning ratio. More in general, we show that there exist graph families such that any straight-line drawing with constant spanning ratio has edge-length ratio which is exponential in the inverse of the toughness. On the other hand, we prove that graph families with constant toughness admit proper straight-line drawings with polynomial edge-length ratio and constant spanning ratio. Finally, we prove that trees with bounded degree admit planar straight-line drawings with polynomial edge-length ratio and constant spanning ratio.
\end{itemize}

\section{Preliminaries} \label{se:preliminaries}

For a graph $G$ and a set $S$ of vertices of $G$, we denote by $G-S$ the graph obtained from $G$ by removing the vertices in $S$ and their incident edges. The subgraph of $G$ \emph{induced by} $S$ is the graph whose vertex set is $S$ and whose edge set consists of every edge of $G$ that has both its end-vertices in $S$. For a vertex $v$, a \emph{$\{v\}$-bridge} of $G$ is the subgraph of $G$ induced by $v$ and by the vertices of a connected component of $G-\{v\}$. The \emph{toughness} of a graph $G$ is the largest real number $t>0$ such that, for every integer $k\geq 2$, $G$ cannot be split into $k$ connected components by the removal of fewer than $t\cdot k$ vertices; that is, for any set $S$ such that $G-S$ consists of $k\geq 2$ connected components, we have $|S|\geq t\cdot k$.  



A \emph{drawing} of a graph maps each vertex to a distinct point in the plane and each edge to a Jordan arc between its end-vertices. A drawing is \emph{straight-line} if it maps each edge to a straight-line segment. Let $\Gamma$ be a straight-line drawing of a graph $G$. The \emph{length of a path} in $\Gamma$ is the sum of the lengths of its edges. We denote by $\|uv\|_{\Gamma}$ (by $\pi_\Gamma(u,v)$) the Euclidean distance (resp.\ the length of a shortest path) between two vertices $u$ and $v$ in $\Gamma$; we sometimes drop the subscript $\Gamma$ when the drawing we refer to is clear from the context. The \emph{spanning ratio} of $\Gamma$ is the real value $\max\limits_{u,v} \frac{\pi_\Gamma(u,v)}{\|uv\|_{\Gamma}}$, where the maximum is over all pairs of vertices $u$ and $v$ of~$G$.

A drawing is \emph{planar} if no two edges intersect, except at common end-vertices. A planar drawing partitions the plane into connected regions, called \emph{faces}; the bounded faces are \emph{internal}, while the unbounded face is the \emph{outer face}. A graph is \emph{planar} if it admits a planar drawing. A planar graph is \emph{maximal} if adding any edge to it violates its planarity. In any planar drawing of a maximal planar graph every face is delimited by a $3$-cycle.

The \emph{bounding box} $\mathcal B(\Gamma)$ of a drawing $\Gamma$ is the smallest axis-parallel rectangle containing $\Gamma$ in the closure of its interior. We denote by $\mathcal B_l(\Gamma)$ and $\mathcal B_r(\Gamma)$ the left and right side of $\mathcal B(\Gamma)$, respectively. The \emph{width} and \emph{height} of $\Gamma$ are the width and height of $\mathcal B(\Gamma)$. 

\section{Drawings with Spanning Ratio $\bf 1$} \label{se:spanning-ratio-one}

In this section we study straight-line drawings with spanning ratio equal to~$1$. We characterize the graphs that admit straight-line drawings, proper straight-line drawings, and planar straight-line drawings with spanning ratio equal to~$1$ and, consequently, derive results on the complexity of recognizing such graphs. We start with the following.

\begin{lemma} \label{le:characterization-general}
	A graph admits a straight-line drawing with spanning ratio equal to $1$ if and only if it contains a Hamiltonian path. 
\end{lemma}

\begin{proof}
	($\Longrightarrow$) Suppose that a graph $G$ admits a straight-line drawing $\Gamma$ with spanning ratio~$1$. Assume, w.l.o.g.\ up to a rotation of the Cartesian axes, that no two vertices of $G$ have the same $x$-coordinate in $\Gamma$. Let $v_1,v_2,\dots,v_n$ be the vertices of $G$, ordered by increasing $x$-coordinates. Then, for $i=1,2,\dots,n-1$, we have that $G$ contains the edge $v_iv_{i+1}$, as otherwise any path between $v_i$ and $v_{i+1}$ would be longer than $\|v_i v_{i+1}\|$. Hence, $G$ contains the Hamiltonian path $(v_1,v_2,\dots,v_n)$.    
	
	($\Longleftarrow$) A straight-line drawing with spanning ratio $1$ of a graph containing a Hamiltonian path $(v_1,v_2,\dots,v_n)$ can be constructed by placing $v_i$ at $(i,0)$, for $i=1,\dots,n$. 
\end{proof}

\begin{theorem} \label{th:complexity-general}
	Recognizing whether a graph admits a straight-line drawing with spanning ratio equal to~$1$ is an NP-complete problem.
\end{theorem}

\begin{proof}
	The theorem follows by Lemma~\ref{le:characterization-general} and from the fact that deciding whether a graph contains a Hamiltonian path is an NP-complete problem~\cite{gj-np-79,gjt-phc-76}.
\end{proof}

A graph $G$ is a \emph{point visibility graph} if there exists a finite point set $P\subset \mathbb R^2$ such that: (i)~$G$ has a vertex for each point in $P$; and (ii) $G$ has an edge between two vertices if and only if the straight-line segment between the corresponding points does not contain any point of~$P$ in its interior; see~\cite[Chapter 15]{d-cg-97}. We have the following.

\begin{lemma} \label{le:characterization-proper}
	A graph admits a proper straight-line drawing with spanning ratio equal to $1$ if and only if it is a point visibility graph. 
\end{lemma}

\begin{proof}
	($\Longrightarrow$) Suppose that a graph $G$ admits a proper straight-line drawing $\Gamma$ with spanning ratio $1$. Let $v_\Gamma$ be the point at which a vertex $v$ of $G$ is drawn in $\Gamma$. Let $P:=\{v_\Gamma \in \mathbb R^2|v\in V(G)\}$ and let $G_P$ be the point visibility graph of $P$. We claim that an edge $uv$ belongs to $G$ if and only if the edge $u_\Gamma v_\Gamma$ belongs to $G_P$; the claim implies that $G_P$ is isomorphic to~$G$ and hence that $G$ is a point visibility graph. 
	First, if $uv$ belongs to $G$, then $\Gamma$ contains the straight-line segment $\overline{u_\Gamma v_\Gamma}$. Since $\Gamma$ is proper, no vertex of $G$ lies in the interior of $\overline{u_\Gamma v_\Gamma}$, hence $G_P$ contains the edge $u_\Gamma v_\Gamma$. Conversely, if $G_P$ contains the edge $u_\Gamma v_\Gamma$, then no point in $P$ lies in the interior of the straight-line segment $\overline{u_\Gamma v_\Gamma}$. Hence, the edge $uv$ belongs to $G$, as otherwise the length of any path between $u$ and $v$ would be larger than $\|u v\|_\Gamma$.
	
	($\Longleftarrow$) Suppose that a graph $G$ is the visibility graph of a point set $P$. For any point $p\in P$, let $v_p$ be the corresponding vertex of $G$. Let $\Gamma$ be the straight-line drawing of $G$ that maps each vertex $v_p$ to the point $p$. Consider any edge $v_pv_q$ of $G$. No vertex $v_r$ lies in the interior of the straight-line segment $\overline{pq}$ in $\Gamma$, as otherwise $v_pv_q$ would not belong to $G$; it follows that $\Gamma$ is proper. Further, consider any two vertices $v_p$ and $v_q$ of $G$ and let $v_p=v_{r_1},v_{r_2},\dots,v_{r_k}=v_q$ be the sequence of vertices of $G$ lying on the straight-line segment $\overline{pq}$ in $\Gamma$, ordered as they occur from $p$ to~$q$. Then $G$ contains the path $(v_p=v_{r_1},v_{r_2},\dots,v_{r_k}=v_q)$, whose length in $\Gamma$ is $\|v_pv_q\|_{\Gamma}$. It follows that the spanning ratio of $\Gamma$ is $1$. 
\end{proof}

The \emph{existential theory of the reals} problem asks whether real values exist for $n$ variables such that a quantifier-free formula, consisting of polynomial equalities and inequalities on such variables, is satisfied. The class of problems that are complete for the existential theory of the reals is denoted by $\exists\mathbb R$~\cite{s-csgtp-09}. It is known that NP $\subseteq\exists\mathbb R\subseteq$ PSPACE~\cite{c-agcp-88}, however it is not known whether $\exists\mathbb R\subseteq$ NP. Many geometric problems are $\exists\mathbb R$-complete, see, e.g.,~\cite{aam-agp-18,m-ut-88}.

\begin{theorem} \label{th:complexity-proper}
	Recognizing whether a graph admits a proper straight-line drawing with spanning ratio equal to~$1$ is an $\exists\mathbb R$-complete problem.
\end{theorem}

\begin{proof}
	The theorem follows by Lemma~\ref{le:characterization-proper} and from the fact that recognizing point visibility graphs is a problem that is $\exists\mathbb R$-complete~\cite{ch-rcpvg-17}.
\end{proof}

We conclude the section by presenting the following.

\begin{theorem} \label{th:complexity-planar}
	Recognizing whether a graph admits a planar straight-line drawing with spanning ratio equal to~$1$ is a linear-time solvable problem.
\end{theorem}
\begin{figure}[htb]\tabcolsep=4pt
	\centering
	\includegraphics[width=.75\textwidth]{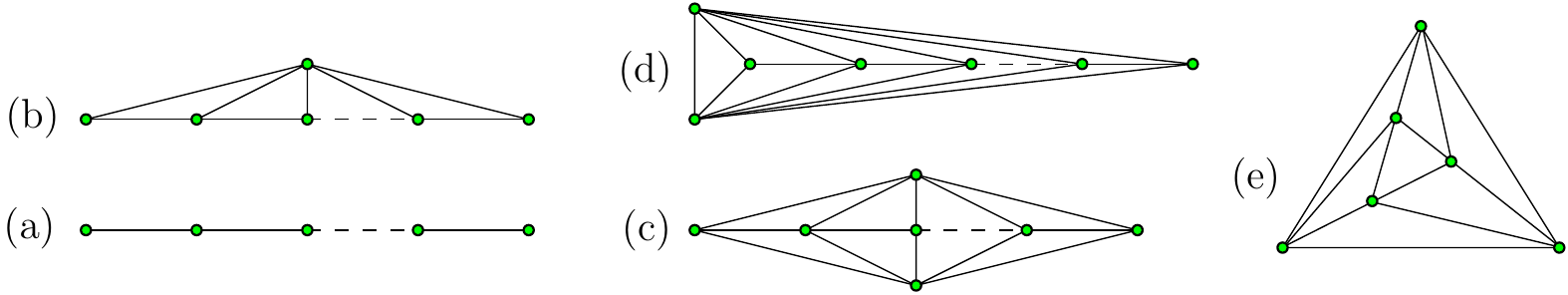}
	\caption{The five graph classes defined in~\cite{desw-dpgfss-07}.}\label{fig:planar-sp1}
\end{figure}
\begin{proof}
	Dujmovi\'c et al.~\cite{desw-dpgfss-07} characterized the graphs that admit a planar straight-line drawing with a straight-line segment between every two vertices as the graphs in the five graph classes in Figure~\ref{fig:planar-sp1}. Since a straight-line drawing has spanning ratio~$1$ if and only if every two vertices are connected by a straight-line segment, the theorem follows from the fact that recognizing whether a graph belongs to such five graph classes can be easily done in linear time.
\end{proof}


\section{Drawings with Spanning Ratio $\mathbf{1 + \epsilon}$} \label{se:planar}

In this section we study straight-line drawings with spanning ratio arbitrarily close to~$1$. Most of the section is devoted to a proof of the following result.

\begin{theorem} \label{th:planar-graph}
	For every $\epsilon>0$, every connected planar graph admits a planar straight-line drawing with spanning ratio smaller than $1+\epsilon$.
\end{theorem} 	

Let $G$ be an $n$-vertex maximal planar graph with $n\geq 3$, let $\mathcal G$ be a planar drawing of $G$, and let $(u,v,z)$ be the cycle delimiting the outer face of $G$ in $\mathcal G$. A \emph{canonical ordering}~\cite{bbc-mco-11,dpp-hdpgg-90,k-dpgco-96}) for $G$ is a total ordering $\sigma_G = [v_1,v_2,\dots,v_n]$ of its vertex set such that the following hold for $k=3,\dots,n$: (i) $v_1=u$, $v_2=v$, and $v_n=z$; (ii) the subgraph $G_k$ of $G$ induced by $v_1,v_2,\dots,v_k$ is $2$-connected and the cycle $\mathcal C_k$ delimiting its outer face in $\mathcal G$ consists of the edge $v_1v_2$ and of a path $\mathcal P_k$ between $v_1$ and $v_2$; and (iii) $v_{k}$ is incident to the outer face of $G_k$ in $\mathcal G$. Theorem~\ref{th:planar-graph} is implied by the following two lemmata.   

\begin{lemma} \label{le:augmentation}
	Let $H$ be any $n$-vertex connected planar graph. There exist an $n$-vertex maximal planar graph $G$ and a canonical ordering $\sigma_G = [v_1,v_2,\dots,v_n]$ for $G$ such that, for each $k\in \{1,2,\dots,n\}$, the subgraph $H_k$ of $H$ induced by $[v_1,v_2,\dots,v_k]$ is connected.
\end{lemma}

\begin{proof}
	For $k=2,3,\dots,n$, let $G_k$ be the subgraph of $G$ induced by $v_1,v_2,\dots,v_k$ and let $L_k$ be the graph composed of $G_k$ and of the vertices and edges of $H$ that are not in $G_k$.
	
	For each $k=2,3,\dots,n$, we define $v_1,v_2,\dots,v_k$ and $G_k$ so that $H_k$ is connected, $G_k$ is $2$-connected, and $L_k$ admits a planar drawing $\mathcal L_k$ such that:
	\begin{enumerate}
		\item the outer face of the planar drawing $\mathcal G_k$ of $G_k$ in $\mathcal L_k$ is delimited by a cycle $\mathcal C_k$ composed of the edge $v_1v_2$ and of a path $\mathcal P_k$ between $v_1$ and $v_2$; 
		\item $v_k$ is incident to the outer face of $\mathcal G_k$; 
		\item every internal face of $\mathcal G_k$ is delimited by a $3$-cycle; and
		\item the vertices and edges of $H$ that are not in $G_k$ lie in the outer face of $\mathcal G_k$. 
	\end{enumerate}
	
	If $k=2$, then construct any planar drawing $\mathcal L_2$ of $H$ and define $v_1$ and $v_2$ as the end-vertices of any edge $v_1v_2$ incident to the outer face of $\mathcal L_2$. Properties~1--4 are then trivially satisfied (in this case the path $\mathcal P_2$ is the single edge $v_1v_2$).
	
	If $2<k<n$, assume that $v_1,v_2,\dots,v_{k-1}$ and $G_{k-1}$ have been defined so that $H_{k-1}$ is connected, $G_{k-1}$ is $2$-connected, and $L_{k-1}$ admits a planar drawing $\mathcal L_{k-1}$ such that Properties~1--4 above are satisfied. Let $\mathcal P_{k-1}=(u=w_1,w_2,\dots,w_x=v)$, where $x\geq 2$. 
	
	Consider any vertex $v$ that is in $L_{k-1}$ and that is not in $G_{k-1}$. By Properties~1 and~4 of $\mathcal L_{k-1}$, all the neighbors of $v$ in $G_{k-1}$ lie in $\mathcal P_{k-1}$. We say that $v$ is a \emph{candidate} (to be designated as $v_k$) \emph{vertex} if, for some $1\leq i\leq x$, there exists an edge $w_i v$ such that $w_i v$ immediately follows the edge $w_i w_{i-1}$ in clockwise order around $w_i$ or immediately follows the edge $w_iw_{i+1}$ in counter-clockwise order around $w_i$; see Figure~\ref{fig:candidate}.
	
	\begin{figure}[htb]\tabcolsep=4pt
		\centering
		\includegraphics[scale=1.4]{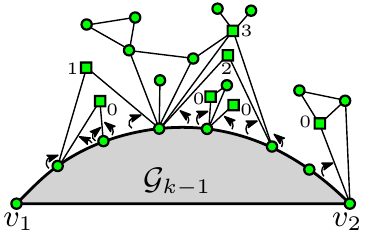}
		\caption{The drawing $\mathcal L_{k-1}$ of $L_{k-1}$, where the interior of $\mathcal G_{k-1}$ is colored gray. Each candidate vertex is represented by a square and labeled with its depth.}
		\label{fig:candidate}
	\end{figure}
	
	For each candidate vertex $v$, let $w_{a(v)}$ and $w_{b(v)}$ be the neighbors of $v$ in $\mathcal P_{k-1}$ such that $a(v)$ is minimum and $b(v)$ is maximum (possibly $a(v)=b(v)$). If $a(v)<b(v)$, define the \emph{reference cycle $\mathcal C(v)$ of} $v$ as the cycle composed of the edges $w_{a(v)} v$ and $w_{b(v)} v$ and of the subpath of $\mathcal P_{k-1}$ between $w_{a(v)}$ and $w_{b(v)}$. Define the \emph{depth} $d(v)$ of $v$ as $0$ if $a(v)=b(v)$ or as the number of candidate vertices that lie inside $\mathcal C(v)$ in $\mathcal L_{k-1}$ otherwise. 
	
	We claim that there exists a candidate vertex with depth $0$. Consider a candidate vertex~$v$ with minimum depth and assume, for a contradiction, that $d(v)>0$; then there exists a candidate vertex $u$ that lies inside $\mathcal C(v)$ in $\mathcal L_{k-1}$. By the planarity of $\mathcal L_{k-1}$, the candidate vertices that lie inside $\mathcal C(u)$ form a subset of those that lie inside $\mathcal C(v)$; moreover, there is at least one candidate vertex, namely $u$, that lies inside $\mathcal C(v)$ and not inside $\mathcal C(u)$, hence $d(u)<d(v)$. This contradicts the assumption that $v$ has \mbox{minimum depth and proves the claim.}
	
	Consider a candidate vertex $v$ with $d(v)=0$. We let $v_k:=v$ and distinguish two cases.
	
	\begin{figure}[htb]
		\centering
		\begin{minipage}[b]{.4\linewidth}
			\centering\includegraphics[scale=1.4]{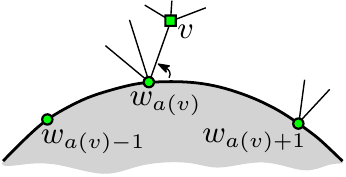}
		\end{minipage}
		\begin{minipage}[b]{.4\linewidth}
			\centering\includegraphics[scale=1.4]{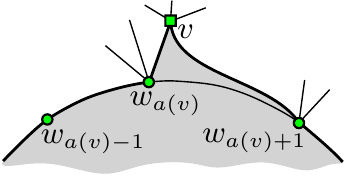}
		\end{minipage}
		\caption{(a) A candidate vertex $v$ with $d(v)=0$ and $a(v)=b(v)$. (b) The drawing $\mathcal L_k$ of $L_k$ obtained by drawing the edge $w_{a(v)+1}v$ in $\mathcal L_{k-1}$.}\label{fig:av=bv}
	\end{figure}

	If $a(v)=b(v)$, assume that $w_{a(v)} v$ immediately follows the edge $w_{a(v)}w_{a(v)+1}$ in counter-clockwise order around $w_{a(v)}$, the other case is symmetric; refer to Figure~\ref{fig:av=bv}. Define $G_k$ as $G_{k-1}$ plus the vertex $v$ and the edges $w_{a(v)}v$ and $w_{a(v)+1}v$. Then $H_k$ is connected because $H_{k-1}$ is connected and the edge $w_{a(v)}v$ belongs to $H$. Further, $G_k$ is $2$-connected because $G_{k-1}$ is $2$-connected and $v$ is adjacent to two distinct vertices of $G_{k-1}$. Define $\mathcal L_k$ by drawing the edge $w_{a(v)+1}v$ so that the cycle $(w_{a(v)},w_{a(v)+1},v)$ does not contain any vertex or edge in its interior. Property~1 is satisfied by $\mathcal L_k$ with $\mathcal P_{k}=(u=w_1,w_2,\dots,w_{a(v)},v,w_{a(v)+1},\dots,w_x=v)$; note that $v$ has no neighbor in $G_{k}$ other than $w_{a(v)}$ and $w_{a(v)+1}$, since $a(v)=b(v)$. Property~2 is satisfied by $\mathcal L_k$ since $\mathcal L_{k-1}$ satisfies Property~4 and by construction. Since the cycle $(w_{a(v)},w_{a(v)+1},v)$ does not contain any vertex in its interior and since $\mathcal L_{k-1}$ satisfies Properties~3 and~4, it follows that $\mathcal L_{k}$ also satisfies Properties~3 and~4. 
	
	\begin{figure}[htb]
		\centering
		\begin{minipage}[b]{.4\linewidth}
			\centering\includegraphics[scale=1.4]{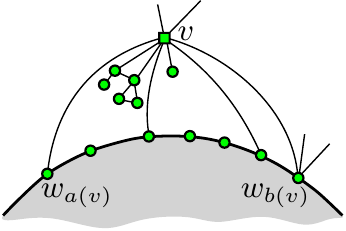}
		\end{minipage}
		\begin{minipage}[b]{.4\linewidth}
			\centering\includegraphics[scale=1.4]{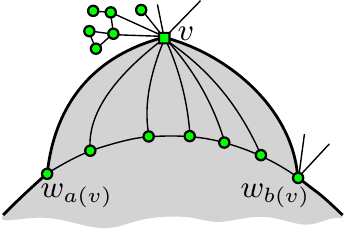}
		\end{minipage}
		\caption{(a) A candidate vertex $v$ with $d(v)=0$ and $a(v)<b(v)$. (b) The drawing $\mathcal L_k$ of $L_k$ obtained by moving out of $\mathcal C(v)$ each $\{v\}$-bridge of $L_{k-1}$ whose vertices different from $v$ lie inside $\mathcal C(v)$ and by drawing the edges among $w_{a(v)}v, w_{a(v)+1}v, \dots, w_{b(v)+1}v$ not in $H$ planarly inside $\mathcal C(v)$.}
		\label{fig:av<bv}
	\end{figure}
	
	Next, we consider the case in which $a(v)<b(v)$; refer to Figure~\ref{fig:av<bv}. We claim that the only edges incident to vertices in the path $(w_{a(v)}, w_{a(v)+1}, \dots, w_{b(v)})$ and lying inside $\mathcal C(v)$ in $\mathcal L_{k-1}$ are those connecting such vertices to $v$. Suppose, for a contradiction, that an edge $w_iu$ with $u\neq v$ lies inside $\mathcal C(v)$. If $a(v)<i<b(v)$, then there exists an edge $w_iz$ with $z\neq v$ that immediately follows $w_iw_{i-1}$ in clockwise order around $w_{i}$ or that immediately follows $w_iw_{i+1}$ in counter-clockwise order around $w_{i}$; hence, $z$ is a candidate vertex. Further, by the planarity of $\mathcal L_{k-1}$, we have that $w_iz$ lies inside $\mathcal C(v)$, except at $w_i$, however this contradicts $d(v)=0$. The proof for the cases in which $i=a(v)$ or $i=b(v)$ is analogous. 
	
	It follows from the previous claim that $v$ is the only vertex of $\mathcal C(v)$ which might have incident edges that lie inside $\mathcal C(v)$ in $\mathcal L_{k-1}$ and that have one end-vertex not in $\mathcal C(v)$. We redraw each $\{v\}$-bridge of $L_{k-1}$ whose vertices different from $v$ lie inside $\mathcal C(v)$ planarly so that it now lies outside $\mathcal C(v)$; after this modification, no vertex of $L_{k-1}$ lies inside $\mathcal C(v)$. 
	
	Define $G_k$ as $G_{k-1}$ plus the vertex $v_k:=v$ and the edges $w_{a(v)}v, w_{a(v)+1}v, \dots, w_{b(v)}v$. Then $H_k$ is connected, because $H_{k-1}$ is connected and the edge $w_{a(v)}v$ belongs to $H$. Further, $G_k$ is $2$-connected, because  $G_{k-1}$ is $2$-connected and $v$ is adjacent to at least two distinct vertices of $G_{k-1}$. Define $\mathcal L_k$ by drawing the edges among $w_{a(v)}v, w_{a(v)+1}v, \dots, w_{b(v)}v$ that do not belong to $H$ so that they all lie inside $\mathcal C(v)$, except at their end-vertices, and so that the edges $w_{a(v)}v, w_{a(v)+1}v, \dots, w_{b(v)}v$ appear consecutively and in this counter-clockwise order around $v$. Property~1 is satisfied by $\mathcal L_k$ with $\mathcal P_{k}=(u=w_1,w_2,\dots,w_{a(v)},v,w_{b(v)},w_{b(v)+1},\dots,w_x=v)$. Property~2 is satisfied by $\mathcal L_k$ by construction and since $\mathcal L_{k-1}$ satisfies Property~4. Every internal face of $\mathcal L_{k}$ that is not an internal face of $\mathcal L_{k-1}$ is delimited by a $3$-cycle $(w_i, w_{i+1},v)$, for some $a(v)\leq i <b(v)$; hence $\mathcal L_{k}$ satisfies Property~3 since $\mathcal L_{k-1}$ does. Finally, $\mathcal L_{k}$ satisfies Property~4 since every vertex or edge of $H$ that is not in $G_k$ lies outside $\mathcal G_{k-1}$ since $\mathcal L_{k-1}$ satisfies Property~4 and lies outside $\mathcal C(v)$ by construction.
	
	If $k=n$, the construction slightly differs from the one described for the case $2<k<n$, as we also require that the outer face of $\mathcal G_n$ is delimited by the $3$-cycle $(v_1,v_2,v_n)$. Hence, if $a(v)=b(v)$ (resp.\ if $a(v)<b(v)$), then $G_n$ also contains the edges $w_1v$, $w_2v$, $\dots$, $w_{a(v)-1}v$, $w_{a(v)+2}v$, $w_{a(v)+3}v$, $\dots$, $w_xv$ (resp.\ the edges $w_1v$, $w_2v$, $\dots$, $w_{a(v)-1}v$, $w_{b(v)+1}v$, $w_{b(v)+2}v$, $\dots$, $w_xv$); further, the edges $w_1v, w_2v, \dots, w_xv$ are drawn in $\mathcal L_n$ in such a way that they appear in this counter-clockwise order around $v$, and so that the outer face of $\mathcal L_n$ is delimited by the $3$-cycle $(w_1=v_1,w_x=v_2,v=v_n)$. The proof that $\mathcal L_n$ satisfies Properties~1--4 is similar, and in fact simpler, than the one described above. 
	
	The above construction implies the statement of the lemma. Namely, $H_k$ is connected for $k=3,4,\dots,n$. Further, $G$ is a maximal planar graph by Property~3 and by the additional requirement for the case $k=n$. We now prove that $[v_1,v_2,\dots,v_n]$ is a canonical ordering for $G$. By Properties~1 and~2 of $\mathcal L_n$, we have that $v_1$, $v_2$, and $v_n$ are incident to the outer face of $\mathcal L_n$; further, for $k=3,4,\dots,n$, we have that $G_k$ is $2$-connected and its outer face in $\mathcal G_k$ is delimited by the edge $v_1v_2$ and by a path $\mathcal P_k$ between $v_1$ and $v_2$, by Property~1 of $\mathcal L_k$; finally, $v_k$ is incident to the outer face of $\mathcal G_k$ for $k=3,4,\dots,n$, by Property~2 of $\mathcal L_k$.   
\end{proof}

\begin{lemma} \label{le:canonical}
	For every $k=3,\dots,n$ and for every $\epsilon>0$, there exists a planar straight-line drawing $\Gamma_k$ of $G_k$ such that: 
	\begin{enumerate}
		\item \label{le:horizontal} the outer face of $\Gamma_k$ is delimited by the cycle $\mathcal C_k$; further, the path $\mathcal P_k$ is $x$-monotone and lies above the edge $uv$, except at $u$ and $v$; and
		\item \label{le:restriction} the restriction $\Xi_k$ of $\Gamma_k$ to the vertices and edges of $H_k$ is a drawing with spanning ratio smaller than $1+\epsilon$.
	\end{enumerate} 
\end{lemma}



\begin{proof}
	The proof is by induction on $k$. If $k=3$, then a planar straight-line drawing $\Gamma_3$ of $G_3$ is constructed by drawing the $3$-cycle $v_1v_2v_3$ as an isosceles triangle in which $v_1v_2$ is horizontal and has length $\epsilon/2$, while $v_1v_3$ and $v_2v_3$ have length $1$, with $v_3$ above the edge $v_1v_2$. By Lemma \ref{le:augmentation}, the graphs $H_2$ and $H_3$ are connected, hence the edge $v_1v_2$ belongs to them and at least one of the edges $v_1v_3$ and $v_2v_3$ belongs to $H_3$. Hence, we have $\frac{\pi(v_1,v_2)}{\|v_1v_2\|}=\frac{\|v_1v_2\|}{\|v_1v_2\|}=1<1+\epsilon$. Further, $\frac{\pi(v_1,v_3)}{\|v_1v_3\|}\leq \max\{\|v_1v_3\|,\|v_1v_2\| + \|v_2v_3\|\}=1+\epsilon/2<1+\epsilon$. Analogously, $\frac{\pi(v_2,v_3)}{\|v_2v_3\|}<1+\epsilon$.
	
	\begin{figure}[htb]
		\tabcolsep=4pt
		\centering
		\includegraphics[scale=1.2]{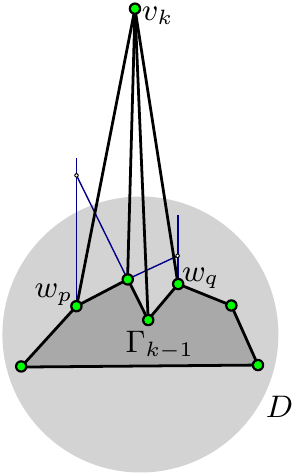}
		\caption{Construction of $\Gamma_k$ from $\Gamma_{k-1}$.}\label{fig:canonical-construction}
	\end{figure}
	
	Now assume that, for some $k=4,\dots,n$, a planar straight-line drawing $\Gamma_{k-1}$ of $G_{k-1}$ has been constructed satisfying Properties~\ref{le:horizontal} and \ref{le:restriction}; refer to Figure~\ref{fig:canonical-construction}. Let $\delta$ be the diameter of a disk $D$ containing $\Gamma_{k-1}$ in its interior. We construct $\Gamma_k$ from $\Gamma_{k-1}$ by placing $v_k$ in the plane as follows. Let $\mathcal P_{k-1}=(u=w_1,w_2,\dots,w_x=v)$. As proved in~\cite{dpp-hdpgg-90}, the neighbors of~$v_k$ in $G_{k-1}$ are the vertices in a sub-path $(w_p,w_{p+1},\dots,w_q)$ of $\mathcal P_{k-1}$, where $1\leq p<q\leq x$. By Property~\ref{le:horizontal} of $\Gamma_{k-1}$, we have $x(w_p)<x(w_{p+1})<\dots<x(w_q)$. We then place $v_k$ at any point in the plane such that the following conditions are satisfied: (i) $x(w_p)<x(v_k)<x(w_q)$; (ii) for every $i=p,\dots,q-1$, the $y$-coordinate of $v_k$ is larger than the $y$-coordinates of the intersection points between the line through $w_iw_{i+1}$ and the vertical lines through $w_p$ and $w_q$; and (iii) the distance between $v_k$ and the point of $D$ closest to $v_k$ is a real value $d>\frac{k \delta}{\epsilon}$.
	
	Since $\mathcal P_{k}$ is obtained from $\mathcal P_{k-1}$ by substituting the path $(w_p,w_{p+1},\dots,w_q)$ with the path $(w_p,v_k,w_q)$, Condition~(i) and the $x$-monotonicity of $\mathcal P_{k-1}$ imply that $\mathcal P_{k}$ is $x$-monotone. Since $\Gamma_{k-1}$ is planar, in order to prove the planarity of $\Gamma_k$ it suffices to prove that no edge incident to $v_k$ intersects any distinct edge of $G_k$, except at common end-vertices. Condition~(ii) implies that the edges incident to $v_k$ lie in the outer face of $\Gamma_{k-1}$, hence they do not intersect any edge of $G_{k-1}$, except at common end-vertices. Again Condition~(ii) and the $x$-monotonicity of $\mathcal P_{k-1}$ imply that no two edges incident to $v_k$ intersect each other, except at $v_k$. We now prove that the spanning ratio of $\Xi_k$ is smaller than $1+\epsilon$. Consider any two vertices $v_i$ and $v_j$. If $i<k$ and $j<k$, then $\frac{\pi_{\Xi_k}(v_i,v_j)}{\|v_iv_j\|_{\Xi_k}}\leq \frac{\pi_{\Xi_{k-1}}(v_i,v_j)}{\|v_iv_j\|_{\Xi_{k-1}}}<1+\epsilon$. If $i=k$, then $\|v_kv_j\|_{\Xi_k}\geq d$, by Condition~(iii). Consider the path $P(v_k,v_j)$ composed of any edge $v_kv_\ell$ in $H_k$ incident to $v_k$ (which exists since $H_k$ is connected) and of any path in $H_{k-1}$ between $v_\ell$ and $v_j$ (which exists since $H_{k-1}$ is connected). The length of $P(v_k,v_j)$ is at most $d+\delta$ (by Condition~(iii) and by the triangular inequality, this is an upper bound on $\|v_kv_\ell\|_{\Xi_k}$) plus $(k-2) \cdot \delta$ (this is an upper bound on the length of any path in $H_{k-1}$). Hence, $\frac{\pi_{\Xi_k}(v_k,v_j)}{\|v_kv_j\|_{\Xi_k}}<\frac{d+k \delta}{d}<1+\epsilon$.
	This completes the induction and the proof of the lemma.
\end{proof}

Lemmata~\ref{le:augmentation} and~\ref{le:canonical} imply Theorem~\ref{th:planar-graph}. Namely, for any connected planar graph $H$, by Lemma~\ref{le:augmentation} we can construct a maximal planar graph $G$ that, by Lemma~\ref{le:canonical} (with $k=n$) and for every $\epsilon>0$, admits a planar straight-line drawing whose restriction to the vertices and edges of $H$ is a drawing with spanning ratio smaller than $1+\epsilon$. 

The following can be obtained by means of techniques similar to (and simpler than) the ones employed in the proof of Theorem~\ref{th:planar-graph}.

\begin{theorem} \label{th:graphs}
	For every $\epsilon>0$, every connected graph admits a proper straight-line drawing with spanning ratio smaller than $1+\epsilon$.
\end{theorem}

\begin{proof}
	Consider any $n$-vertex graph $G$ and let $T$ be any spanning tree of $G$. Let $v_1,v_2,\dots,v_n$ be any total ordering for the vertex set of $T$ such that the subtree $T_k$ of $T$ induced by $v_1,v_2,\dots,v_k$ is connected, for each $k=1,2,\dots,n$. 
	
	For $k=1,2,\dots,n$, we construct a straight-line drawing $\Gamma_k$ of $T_k$ with spanning ratio smaller than $1+\epsilon$ and such that no three vertices lie on a straight line. If $k=1$, then $\Gamma_1$ is constructed by placing $v_1$ at any point in the plane. Now assume that a straight-line drawing $\Gamma_{k-1}$ of $T_{k-1}$ has been constructed with spanning ratio smaller than $1+\epsilon$ and such that no three vertices lie on a straight line. Let $\delta$ be the diameter of a disk $D$ containing $\Gamma_{k-1}$ in its interior. We construct $\Gamma_k$ from $\Gamma_{k-1}$ by placing $v_k$ at any point in the plane such that: (1) $v_k$ does not lie on any straight line through two vertices of $T_{k-1}$; and (2) the distance between $v_k$ and the point of $D$ that is closest to $v_k$ is a real value $d>\frac{k \delta}{\epsilon}$.
	
	By Property~(1), no three vertices lie on a straight line in $\Gamma_k$. We prove that the spanning ratio of $\Gamma_k$ is smaller than $1+\epsilon$. Consider any two vertices $v_i$ and $v_j$. If $i<k$ and $j<k$, then $\frac{\pi_{\Gamma_k}(v_i,v_j)}{\|v_iv_j\|_{\Gamma_k}}\leq \frac{\pi_{\Gamma_{k-1}}(v_i,v_j)}{\|v_iv_j\|_{\Gamma_{k-1}}}<1+\epsilon$. If $i=k$, then $\|v_kv_j\|_{\Gamma_k}\geq d$, by Property~(2). Further, $\pi_{\Gamma_k}(v_k,v_j)$ is at most $d+\delta$ (by Property~(2) and by the triangular inequality, this is an upper bound on the length of the edge of $T_k$ incident to $v_k$) plus $(k-2) \cdot \delta$ (this is an upper bound on the length of any path in $T_{k-1}$). Hence, $\frac{\pi_{\Gamma_k}(v_k,v_j)}{\|v_kv_j\|_{\Gamma_k}}<\frac{d+k \delta}{d}<1+\epsilon$.

	A drawing $\Gamma$ of $G$ is obtained from the drawing $\Gamma_n$ of $T=T_n$ by drawing the edges that are not in $T$ as straight-line segments. Then $\Gamma$ is proper, as no three vertices of $T$ lie on a straight line in $\Gamma_n$, and has spanning ratio smaller than $1+\epsilon$, as the same is true for $\Gamma_n$.
\end{proof}

%

\section{Drawings with Small Spanning Ratio and Edge-Length Ratio} \label{se:area}

In this section we study straight-line drawings with small spanning ratio and edge-length ratio. Our main result is the following.


\begin{theorem} \label{th:tough} 
	For every $\epsilon>0$ and $\tau>0$, every $n$-vertex graph with toughness~$\tau$ admits a proper straight-line drawing whose spanning ratio is at most $1+\epsilon$ and whose edge-length ratio is in $\mathcal O\left(n^{\frac{\log_2 (2+ \lceil 2/\epsilon \rceil)}{\log_2(2+\lceil 1/\tau \rceil)-\log_2(1+\lceil 1/\tau \rceil)}}\cdot 1/\epsilon\right)$.
	
	Further, for every $0<\tau<1$, there is a graph $G$ with toughness $\tau$ such that every straight-line drawing of $G$ with spanning ratio at most $s$ has edge-length ratio in $2^{\Omega(1/(\tau\cdot s^2))}$.
\end{theorem}

In order to prove Theorem~\ref{th:tough}, we study straight-line drawings of bounded-degree trees. This is because there is a strong connection between the toughness of a graph and the existence of a spanning tree with bounded degree. Indeed, if a graph $G$ has toughness $\tau$, then it has a spanning tree with maximum degree $\lceil 1/\tau \rceil+2$~\cite{w-cbet-89}. Further, a tree has toughness equal to the inverse of its maximum degree. We start by proving the following upper bound.

\begin{theorem} \label{th:trees-upperbound-nonplanar} 
	For every $\epsilon>0$, every $n$-vertex tree $T$ with maximum degree $d$ admits a proper straight-line drawing such that no three vertices are collinear, the spanning ratio is at most $1+\epsilon$, the distance between any two vertices is at least $1$, and the width, the height, and the edge-length ratio are in $\mathcal O\left(n^{\frac{\log_2 (2+ \lceil 2/\epsilon \rceil)}{\log_2(d/(d-1))}}\cdot 1/\epsilon\right)$.
\end{theorem}

\begin{proof}
	Let $\gamma=\lceil \frac{2}{\epsilon}\rceil $. Root $T$ at any vertex $r$. For any two vertices $p$ and $q$ of $T$, let $P_{pq}$ be the path in $T$ from $p$ to $q$. We prove that, for an arbitrary real value $\eta>0$, $T$ admits a proper straight-line drawing $\Gamma$ that, in addition to the properties in the statement of the theorem, satisfies the following: (1) $r$ is at the top-left corner of $\mathcal B(\Gamma)$; (2) for every vertex $z$ of $T$, the path $P_{zr}$ monotonically decreases in the $x$-direction and monotonically increases in the $y$-direction from $z$ to $r$; and (3) the height of $\Gamma$ is at most $\eta$. 
	
	If $n=1$, then $\Gamma$ is obtained by placing $r$ at any point in the plane. If $n>1$, then there exists an edge $uv$ whose removal separates $T$ into two trees $T_1$ and $T_2$, each with at most $\frac{d-1}{d}n$ vertices~\cite{v-ucvc-81}. Assume, w.l.o.g., that $T_1$ contains $r$ and $u$, while $T_2$ contains $v$. Then $T_1$ is rooted at $r$ and $T_2$ is rooted at $v$. Inductively construct proper straight-line drawings $\Gamma_1$ and $\Gamma_2$ of $T_1$ and $T_2$, respectively, with parameter $\eta/3$ satisfying Properties~1--4. Let $w_1$ and $w_2$ be the widths of $\Gamma_1$ and $\Gamma_2$, respectively. Refer to Figure~\ref{fig:trees-nonplanar}.

	\begin{figure}[htb]\tabcolsep=4pt
		\centering
		\includegraphics[scale=1.2]{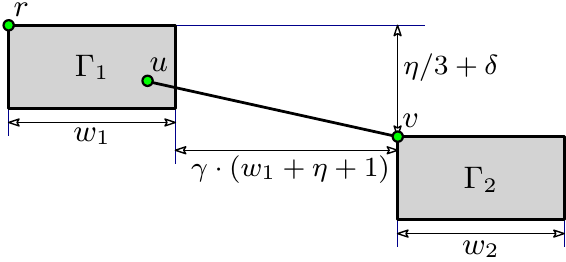}
		\caption{Illustration for the construction in Theorem~\ref{th:trees-upperbound-nonplanar}.}
		\label{fig:trees-nonplanar}
	\end{figure}
	
	Translate $\Gamma_1$ so that $r$ lies at $(0,0)$. Further, translate $\Gamma_2$ so that $v$ lies at $(w_1+\gamma\cdot (w_1+\eta+1),-\eta/3-\delta)$, where $0<\delta<\eta/3$ is a real value chosen so that no line through two vertices in the same tree $T_i$, with $i\in\{1,2\}$, overlaps a vertex in the tree $T_j$, with $j\in\{1,2\}$ and $j\neq i$. Note that, since $\Gamma_1$ and $\Gamma_2$ are proper and since $\mathcal B(\Gamma_1)$ and $\mathcal B(\Gamma_2)$ are disjoint, there are finitely many values of $\delta$ for which the line through two vertices in a tree $T_i$ overlaps a vertex in a different tree $T_j$, hence such a value $\delta$ always exists. 
	
	We now analyze the properties of $\Gamma$. By construction, $\Gamma$ is a straight-line drawing of $T$. 
	
	By induction, no three vertices are collinear in each of $\Gamma_1$ and $\Gamma_2$; further, by construction, $\Gamma_1$ and $\Gamma_2$ are arranged so that no line through two vertices in the same tree $T_i$, with $i\in\{1,2\}$, overlaps a vertex in the tree $T_j$, with $j\in\{1,2\}$ and $j\neq i$. Hence, no three vertices are collinear in $\Gamma$, and in particular.  $\Gamma$ is proper.
	
	Property~(1) is satisfied by $\Gamma$ given that $r$ is at the top-left corner of $\mathcal B(\Gamma_1)$, by induction, and given that every vertex of $T_2$ lies to the right and below $r$ in $\Gamma$, by construction. 
	
	In order to prove that $\Gamma$ satisfies Property~(2), consider any vertex $z$ of $T$. If $z$ belongs to $T_1$, then $P_{zr}$ monotonically decreases in the $x$-direction and monotonically increases in the $y$-direction from $z$ to $r$, since $\Gamma_1$ satisfies Property~(2). If $z$ belongs to $T_2$, then $P_{zr}$ is composed of the path $P_{zv}$, of the edge $vu$, and of the path $P_{ur}$. The paths $P_{zv}$ and $P_{ur}$ monotonically decrease in the $x$-direction and monotonically increase in the $y$-direction from $z$ to $v$ and from $u$ to $r$, respectively, since $\Gamma_2$ and $\Gamma_1$ satisfy Property~(2). Further, the $x$-coordinate of $u$ is smaller than the one of $v$ and the $y$-coordinate of $u$ is larger than the one of $v$; the latter follows from the fact that every vertex of $T_1$ has $y$-coordinate in $[-\eta/3,0]$, while every vertex of $T_2$ has $y$-coordinate smaller than $-\eta/3$. 
	
	The height of $\Gamma$ is at most $2\eta/3 + \delta$, which is smaller than $\eta$, hence $\Gamma$ satisfies Property~(3). 
	
	We now discuss the spanning ratio of $\Gamma$. We prove that, for any vertex $w$ of $T_1$ and any vertex $z$ of $T_2$, it holds true that $\frac{\pi_\Gamma(w,z)}{\|wz\|_\Gamma}\leq \frac{\gamma+2}{\gamma}$. This suffices to prove that the spanning ratio of $\Gamma$ is at most $\frac{\gamma+2}{\gamma}$, since the drawings of $T_1$ and $T_2$ in $\Gamma$ are the ones inductively constructed by the algorithm. The distance between $w$ and $z$ is larger than or equal to $\gamma\cdot (w_1 + \eta+1)+x_z$, where $\gamma\cdot (w_1 + \eta+1)$ is the horizontal distance between $\mathcal B_r(\Gamma_1)$ and $\mathcal B_l(\Gamma_2)$, while $x_z$ denotes the distance between $\mathcal B_l(\Gamma_2)$ and $z$.
	Clearly, we have $\pi_\Gamma(w,z) = \pi_\Gamma(w,r) + \pi_\Gamma(z,r)$. The path $P_{wr}$ is monotone in the $x$- and $y$-directions, hence $\pi_\Gamma(w,r)$ is upper bounded by the horizontal distance between $w$ and $r$, which is at most $w_1$, plus the vertical distance between $w$ and $r$, which is at most $\eta/3$. Analogously, $\pi_\Gamma(z,r)$ is upper bounded by the horizontal distance between $z$ and $r$, which is $w_1+\gamma\cdot (w_1 + \eta +1) +x_z$, plus the vertical distance between $z$ and $r$, which is at most $\eta$. Hence, $\pi_\Gamma(w,z)< (\gamma+2) \cdot (w_1+\eta+1)+x_z$. Thus:
	
	$$\frac{\pi_\Gamma(u,v)}{\|uv\|_\Gamma}< \frac{(\gamma+2) \cdot (w_1+\eta+1+\frac{x_z}{\gamma})}{\gamma\cdot (w_1 + \eta+1+\frac{x_z}{\gamma})}\leq \frac{\gamma+2}{\gamma}\leq 1+\epsilon.$$
	
	Finally, we analyze the edge-length ratio of $\Gamma$. Note that the distance between any vertex of $T_1$ and any vertex of $T_2$ is larger than $1$, hence the same is true for every pair of vertices of $T$. In particular, the length of every edge is larger than $1$. Thus, the edge-length ratio of $\Gamma$ is upper bounded by the maximum length of an edge of $T$. In turn, this is at most the height plus the width of $\Gamma$. By Property~(3), the height of $\Gamma$ is at most $\eta$. By construction, the width of $\Gamma$ is equal to $w_1 + \gamma\cdot (w_1 + \eta +1) +w_2$. Denote by $w(n)$ the maximum width of a drawing of an $n$-vertex tree constructed by the above algorithm. Since each of $T_1$ and $T_2$ has at most $\frac{d-1}{d}n$ vertices, we get that $w(n)\leq (\gamma+2) \cdot (w(\frac{d-1}{d}n)+\eta+1)$. Repeatedly substituting this inequality into itself and recalling that $w(n)=0$ for $n\leq 1$, we get $w(n)\leq (1+\eta)\cdot(\gamma+2)  + (1+\eta)\cdot(\gamma+2)^2 + \dots +  (1+\eta)\cdot(\gamma+2)^{\lceil\log_{\frac{d}{d-1}}(n)\rceil}\leq (1+\eta)\cdot\frac{\gamma+2}{\gamma+1}\cdot(\gamma+2)^{\log_{\frac{d}{d-1}}(n)+1}=(1+\eta)\cdot\frac{\gamma+2}{\gamma+1}\cdot(\gamma+2)\cdot n^{\frac{\log_2 (\gamma+2)}{\log_2(d/(d-1))}}$. We have $\frac{\gamma+2}{\gamma+1}\leq 2$, given that $\gamma=\lceil \frac{2}{\epsilon}\rceil\geq 1$; further, we can set $\eta$ to be any constant, say $\eta=1$. Thus, we get $w(n)\in \mathcal O\left(n^{\frac{\log_2 (2+ \lceil 2/\epsilon \rceil)}{\log_2(d/(d-1))}}\cdot 1/\epsilon\right)$ and the same holds true for the edge-length ratio of $\Gamma$. 
\end{proof}

We can now prove the upper bound in Theorem~\ref{th:tough}. Consider an $n$-vertex graph $G$ with toughness $\tau$ and let $\epsilon>0$; then $G$ has a spanning tree $T$ with maximum degree $d=\lceil 1/\tau \rceil+2$~\cite{w-cbet-89}. Apply Theorem~\ref{th:trees-upperbound-nonplanar} in order to construct a straight-line drawing $\Gamma_T$ of $T$. Construct a straight-line drawing $\Gamma_G$ of $G$ from $\Gamma_T$ by representing the edges of $G$ not in $T$ as straight-line segments. This results in a proper drawing of $G$, given that no three vertices are collinear in $\Gamma_T$. Further, the spanning ratio of $\Gamma_G$ is at most the one of $\Gamma_T$, hence it is at most $1+\epsilon$. Finally, the edge-length ratio of $\Gamma_G$ is in $\mathcal O\left(n^{\frac{\log_2 (2+ \lceil 2/\epsilon \rceil)}{\log_2(d/(d-1))}}\cdot 1/\epsilon\right)$, given that the distance between any two vertices in $\Gamma_T$ (and hence in $\Gamma_G$) is at least $1$ and given that the width and height of $\Gamma_T$ (and hence of $\Gamma_G$) are in $\mathcal O\left(n^{\frac{\log_2 (2+ \lceil 2/\epsilon \rceil)}{\log_2(d/(d-1))}}\cdot 1/\epsilon\right)$. Substituting the value $d=\lceil 1/\tau \rceil+2$ provides us with the upper bound in Theorem~\ref{th:tough}. 

The lower bound in Theorem~\ref{th:tough} comes from the following theorem.

\begin{theorem} \label{th:lower}
	Let $T$ be a tree with a vertex of degree $d$. For any $s\geq 1$, any straight-line drawing of $T$ with spanning ratio at most $s$ has edge-length ratio in $2^{\Omega(d/s^2)}$.
\end{theorem} 

\begin{proof}
	For any $s\geq 1$, let $\Gamma$ be any straight-line drawing of $T$ with spanning ratio at most~$s$; refer to Figure~\ref{fig:lower}. Let $u_T$ be a vertex of degree $d$. Assume w.l.o.g.\ up to a scaling (which does not alter the edge-length ratio and the spanning ratio of $\Gamma$) that the length of the shortest edge incident to $u_T$ in $\Gamma$ is $1$. For any integer $i\geq 0$, let $\mathcal C_i$ be the circle centered at $u_T$ whose radius is $r_i=2^i$. Further, for any integer $i>0$, let $\mathcal A_i$ be the closed annulus delimited by $\mathcal C_{i-1}$ and $\mathcal C_{i}$. By assumption, no neighbor of $u_T$ lies inside the open disk delimited by $\mathcal C_0$. We claim that, for any integer $i>0$ and for some constant $c$, there are at most $c\cdot s^2$ neighbors of $u_T$ inside $\mathcal A_i$. This implies that at most $k\cdot c\cdot s^2$ neighbors of $u_T$ lie inside the closed disk delimited by $\mathcal C_{k}$. Hence, if $d>k\cdot c\cdot s^2$, e.g., if $k=\lfloor \frac{d-1}{c\cdot s^2}\rfloor$, then there is a neighbor $v_T$ of $u_T$ outside $\mathcal C_{k}$. Then $\|u_Tv_T\|>2^{k}\in 2^{\Omega(d/s^2)}$. Hence, the theorem follows from the claim.
	
	\begin{figure}[htb]\tabcolsep=4pt
		\centering
		\includegraphics[scale=1.1]{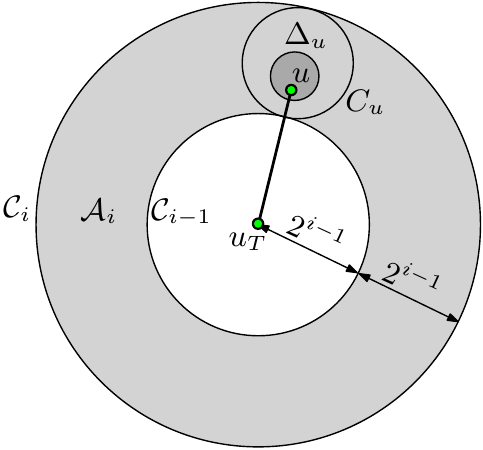}
		\caption{Illustration for the proof of Theorem~\ref{th:lower}.}
		\label{fig:lower}
	\end{figure}
	
	It remains to prove the claim. For each neighbor $u$ of $u_T$ inside $\mathcal A_i$, let $\Delta_u$ be a closed disk such that: (i) $u$ lies inside $\Delta_u$; (ii) $\Delta_u$ lies inside $\mathcal A_i$; and (iii) the diameter of $\Delta_u$ is $\delta_i=2^{i-2}/s$. The existence of $\Delta_u$ can be proved as follows. Consider the circle $C_u$ whose antipodal points are the intersection points of $\mathcal C_{i-1}$ and $\mathcal C_{i}$ with the ray from $u_T$ through $u$. Note that $C_u$ lies inside $\mathcal A_i$ and has diameter $2^{i-1}>\delta_i=2^{i-2}/s$. Then $\Delta_u$ is any disk with diameter $\delta_i$ that contains $u$ and that lies inside the closed disk delimited by $C_u$.
	
	Suppose, for a contradiction, that there exist two neighbors $u$ and $v$ of $u_T$ inside $\mathcal A_i$ such that the disks $\Delta_u$ and $\Delta_v$ intersect. Then $\pi_\Gamma(u,v)\geq 2^i$, since both the edges $uu_T$ and $vu_T$ are longer than $r_{i-1}=2^{i-1}$. By the triangular inequality, $||uv||_\Gamma\leq 2\cdot \delta_i=2^{i-1}/s$. Hence $\frac{\pi_\Gamma(u,v)}{||uv||_\Gamma}\geq 2s$, while the spanning ratio of $\Gamma$ is at most $s$. This contradiction proves that, for any two neighbors $u$ and $v$ of $u_T$ inside $\mathcal A_i$, the disks $\Delta_u$ and $\Delta_v$ do not intersect.
	
	The area of $\mathcal A_i$ is $\pi\cdot (r_i^2-r_{i-1}^2)=\pi\cdot (2^{2i}-2^{2i-2})=3\pi\cdot (2^{2i-2})$. Since each disk $\Delta_u$ lying inside $\mathcal A_i$ has area $\pi\cdot (2^{2i-6}/s^2)$ and does not intersect any different disk $\Delta_v$, it follows that $\mathcal A_i$ contains at most $\frac{3\pi\cdot (2^{2i-2})\cdot s^2}{\pi\cdot 2^{2i-6}}=48\cdot s^2$ distinct disks $\Delta_u$ and hence at most $48\cdot s^2$ neighbors of $u_T$. This proves the claim and concludes the proof of the theorem.
\end{proof}

\begin{corollary}
	Let $S$ be an $n$-vertex star. For any $s\geq 1$, any straight-line drawing of $S$ with spanning ratio at most $s$ has edge-length ratio in $2^{\Omega(n/s^2)}$.
\end{corollary}

The lower bound of Theorem~\ref{th:tough} follows from Theorem~\ref{th:lower} and from the fact that a tree with maximum degree $d$ has toughness $1/d$. This concludes the proof of Theorem~\ref{th:tough}.

We now prove that trees with bounded maximum degree admit planar straight-line drawings with constant spanning ratio and polynomial edge-length ratio. The cost of planarity is found in the dependence on the maximum degree, \mbox{which is worse than in Theorem~\ref{th:trees-upperbound-nonplanar}.}

\begin{theorem} \label{th:trees-upperbound} 
	For every $\epsilon>0$, every $n$-vertex tree $T$ with maximum degree $d$ admits a planar straight-line drawing whose spanning ratio is at most $1+\epsilon$ and whose edge-length ratio is in $\mathcal O\left((2n)^{2+(d-2)\cdot \log_2( 1+\lceil\frac{2}{\epsilon}\rceil)}\cdot  \log_2 n\right)$.
\end{theorem}

\begin{proof}
	Let $\gamma=\lceil \frac{2}{\epsilon} \rceil$. If $d\leq 2$, then $T$ is a path and a planar straight-line drawing with spanning ratio $1$ and edge-length ratio $1$ is trivially constructed. We can hence assume that $d\geq 3$. Root $T$ at any leaf $r$; this ensures that every vertex of $T$ has at most $d-1$ children. 
	In order to avoid some technicalities in the upcoming algorithm, we also assume that every non-leaf vertex of $T$ has at least two children. This is obtained by inserting a new child for each vertex of $T$ with just one child; note that the \emph{size} of the tree, i.e., its number of vertices, is less than doubled by this modification. We again call $T$ the tree after this modification and by $n$ its size. Clearly, no path between two vertices of the initial tree uses the newly inserted vertices, hence removing the inserted vertices together with their incident edges from a drawing with spanning ratio smaller than or equal to $1+\epsilon$ of the modified tree results in a drawing with spanning ratio smaller than or equal to $1+\epsilon$ of the initial tree.
	
	Our construction is a ``well-spaced'' version of an algorithm by Shiloach~\cite{s-lpag-76}. Namely, we construct a planar straight-line drawing $\Gamma$ of $T$ in which (i) $r$ is at the top-left corner of $\mathcal B(\Gamma)$, and (ii) for every vertex $u$ of $T$, the path from $u$ to $r$ in $T$ is (non-strictly) $xy$-monotone. 
	
	If $n=1$, then $\Gamma$ is obtained by placing $r$ at any point in the plane. 
	If $n>1$, then let $r_1,r_2,\dots,r_k$ be the children of $r$, where $k\leq d-1$, let $T_1,T_2,\dots,T_k$ be the subtrees of $T$ rooted at $r_1,r_2,\dots,r_k$, and let $n_1,n_2,\dots,n_k$ be the size of $T_1,T_2,\dots,T_k$, respectively. Assume, w.l.o.g.\ up to a relabeling, that $n_1\leq n_2 \leq \dots \leq n_k$; hence, $n_i\leq n/2$ for $i=1,2,\dots,k-1$. Refer to Figure~\ref{fig:trees}. Place $r$ at any point in the plane. Inductively construct planar straight-line drawings $\Gamma_1,\Gamma_2,\dots,\Gamma_k$ of $T_1,T_2,\dots,T_k$, respectively. 
	Position $\Gamma_1$ so that $r_1$ is on the same vertical line as $r$, one unit below it; let $d_1$ be the width of $\Gamma_1$. Then, for $i=2,\dots,k$, position $\Gamma_i$ so that $r_i$ is one unit below $r$ and $\gamma\cdot (d_{i-1}+ \log_2 n)$ units to the right of $\mathcal B_r(\Gamma_{i-1})$; denote by $d_i$ the width of the bounding box of the drawings $\Gamma_1,\Gamma_2,\dots,\Gamma_i$. Finally, move $\Gamma_k$ one unit above, so that $r_k$ is on the same horizontal line as $r$.

	\begin{figure}[htb]\tabcolsep=4pt
		\centering
		\includegraphics[scale=1.2]{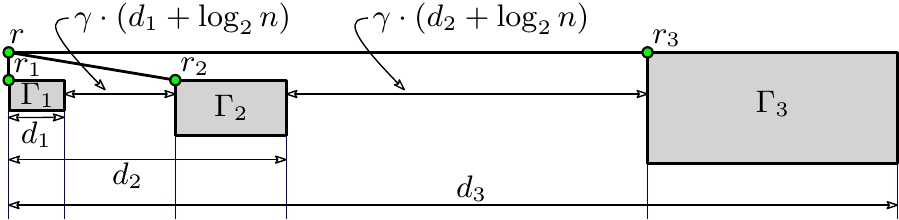}
		\caption{Inductive construction of $\Gamma$. In this example $k=3$.}
		\label{fig:trees}
	\end{figure}
	
	We now analyze the properties of $\Gamma$. By construction, we have that $\Gamma$ is a straight-line drawing. The planarity of $\Gamma$ is easily proved by exploiting the fact that $r_i$ is at the top-left corner of $\mathcal B(\Gamma_i)$ and that $r_1,r_2,\dots,r_{k-1}$ all lie one unit below $r$.
	
	{\bf Height.} Denote by $h(n)$ the maximum height of a drawing of an $n$-vertex tree constructed by the above algorithm. The same analysis as in~\cite{s-lpag-76} shows that $h(n)\leq \log_2 n$. This comes from $h(1)=0$ and $h(n)\leq \max\{h(\frac{n}{2})+1,h(n-1)\}$ for $n\geq 2$.
	
	{\bf Spanning ratio.} We prove that, for any two vertices $u$ and $v$ that do not belong to the same subtree $T_i$, it holds true that $\frac{\pi_\Gamma(u,v)}{\|uv\|_\Gamma}\leq \frac{\gamma+2}{\gamma}$. This suffices to prove that the spanning ratio of $\Gamma$ is at most $\frac{\gamma+2}{\gamma}$. Suppose that $u$ belongs to a subtree $T_i$ and $v$ belongs to a subtree $T_j$, with $i<j$; the case in which one of $u$ and $v$ is $r$ can be discussed analogously. 
	
	First, we have $\|uv\|\geq x_v + \gamma\cdot (d_{j-1}+ \log_2 n)$, where $x_v$ denotes the distance between $v$ and $\mathcal B_l(\Gamma_j)$, while the second term is the distance between $\mathcal B_l(\Gamma_j)$ and $\mathcal B_r(\Gamma_{j-1})$. 
	
	Clearly, we have $\pi_\Gamma(u,v) = \pi_\Gamma(u,r) + \pi_\Gamma(r,v)$. The path between $u$ and $r$ (between $v$ and $r$) is $xy$-monotone, hence $\pi_\Gamma(u,r)$ (resp.\ $\pi_\Gamma(v,r)$) is upper bounded by the horizontal distance plus the vertical distance between $u$ and $r$ (resp.\ between $v$ and $r$). The vertical distance between $u$ and $r$ (between $v$ and $r$) is at most $\log_2(n)$, since the height of $\Gamma$ is at most $\log_2(n)$. The horizontal distance between $u$ and $r$ is at most $d_i\leq d_{j-1}$, while the one between $v$ and $r$ is $x_v + \gamma\cdot (d_{j-1}+  \log_2 n) + d_{j-1}$. Hence, $\pi_\Gamma(u,v)\leq (d_{j-1}+ \log_2 n) + (x_v + \gamma\cdot (d_{j-1}+ \log_2 n) + d_{j-1} + \log_2 n )  = x_v + (\gamma+2) \cdot (d_{j-1}+ \log_2 n)$. Thus:
	\begin{eqnarray*}
		\frac{\pi_\Gamma(u,v)}{\|uv\|_\Gamma}\leq \frac{(\gamma+2) \cdot (\frac{x_v}{\gamma} + d_{j-1}+ \log_2 n)}{ \gamma\cdot (\frac{x_v}{\gamma} + d_{j-1}+ \log_2 n)}\leq \frac{\gamma+2}{\gamma}\leq 1+\epsilon.
	\end{eqnarray*}	
	
	{\bf Width.} Let $w_1,\dots,w_k$ be the widths of $\Gamma_1,\dots,\Gamma_k$. By construction, $d_1=w_1$ and, for each $j=2,\dots,k$, we have $d_j=d_{j-1} + \gamma\cdot (d_{j-1}+ \log_2 n) + w_j=(\gamma+1)\cdot d_{j-1} + \gamma\cdot \log_2 n + w_j$. Hence, by induction on $j$, we have $d_j=(\gamma+1)^{j-1} \cdot w_1 + (\gamma+1)^{j-2} \cdot w_2 + \ldots + (\gamma+1) \cdot  w_{j-1}+w_j + ((\gamma+1)^{j-1}-1) \cdot \log_2 n$.
	In particular, the width of $\Gamma$ is equal to $d_k$ and hence to:
	\begin{eqnarray} \label{eq:k>1}
	\sum\limits_{i=1}^{k} ((\gamma+1)^{k-i} \cdot w_i) + ((\gamma+1)^{k-1}-1) \cdot \log_2 n .
	\end{eqnarray} 
	
	For the reminder of the proof, we introduce the notation $k_1=k$ and $n_{1,i}=n_i$, for $i=1,2,\dots,k_1$. Recall that $k_1\leq d-1$. Denote by $w(n)$ the maximum width of a drawing of an $n$-vertex tree constructed by the above algorithm. By construction, we have $w(1)=0$. For $n\geq 2$, by Equality~\ref{eq:k>1}, we get:
	\begin{eqnarray} \label{eq:degree}
	w(n)\leq (\gamma+1)^{d-2} \cdot \sum\limits_{i=1}^{k_1-1} w(n_{1,i}) + w(n_{1,k_1}) + (\gamma+1)^{d-2} \cdot \log_2 n.
	\end{eqnarray} 
	
	Let $r_{2,1},r_{2,2},\dots,r_{2,k_2}$ be the children of $r_k$, where $k_2\leq d-1$, and let $n_{2,1},n_{2,2},\dots,n_{2,k_2}$ be size of the subtrees $T_{2,1},T_{2,2},\dots,T_{2,k_2}$ of $T$ rooted $r_{2,1},r_{2,2},\dots,r_{2,k_2}$, respectively. Assume, w.l.o.g., that $n_{2,1}\leq n_{2,2} \leq \dots \leq n_{2,k_2}$; hence, $n_{2,i}\leq n/2$ for $i=1,2,\dots,k_2-1$. By the same argument used to derive Inequality~\ref{eq:degree}, we get that the term $w(n_{1,k_1})$ in Inequality~\ref{eq:degree} can be replaced by $(\gamma+1)^{d-2} \cdot \sum\limits_{i=1}^{k_2-1} w(n_{2,i}) + w(n_{2,k_2}) + (\gamma+1)^{d-2} \cdot \log_2 n$, hence we get
	\begin{eqnarray} \label{eq:degree-2}
	w(n)\leq (\gamma+1)^{d-2} \cdot \left(\sum\limits_{i=1}^{k_1-1} w(n_{1,i}) + \sum\limits_{i=1}^{k_2-1} w(n_{2,i})\right) + w(n_{2,k_2}) + 2\cdot (\gamma+1)^{d-2} \cdot \log_2 n.
	\end{eqnarray} 
	
	Again, the term $w(n_{2,k_2})$ in Inequality~\ref{eq:degree-2} can be replaced by $(\gamma+1)^{d-2} \cdot \sum\limits_{i=1}^{k_3-1} w(n_{3,i}) + w(n_{3,k_3}) + (\gamma+1)^{d-2} \cdot \log_2 n$, where $n_{3,1},n_{3,2},\dots,n_{3,k_3}$ are the sizes of the subtrees of $T$ rooted at the children of $r_{2,k_2}$, with $k_3\leq d-1$ and  $n_{3,i}\leq n/2$ for $i=1,2,\dots,k_3-1$. The repetition of this argument eventually leads to the inequality:
	\begin{eqnarray*} 
		w(n)\leq (\gamma+1)^{d-2} \cdot \left(\sum\limits_{i=1}^{k_1-1} w(n_{1,i}) + \sum\limits_{i=1}^{k_2-1} w(n_{2,i}) +  \dots + \sum\limits_{i=1}^{k_t-1} w(n_{t,i})\right) + t\cdot (\gamma+1)^{d-2} \cdot \log_2 n,
	\end{eqnarray*} 
	
	\noindent where $t$ is the index at which $r_{t,k_t}$ has no children, hence $n_{t,k_t}=1$ and $w(n_{t,k_t})=0$. Since $t\leq n-1$, we get
	\begin{eqnarray} \label{eq:substitution}
	w(n)\leq (\gamma+1)^{d-2} \cdot \sum\limits_{i,j} w(n_{j,i}) + (\gamma+1)^{d-2} \cdot (n-1)\cdot \log_2 n,
	\end{eqnarray} 
	
	\noindent where the sum is defined over all pair of integers $j$ and $i$ such that $1\leq j \leq t$ and $1\leq i\leq k_j-1$. 
	
	We prove, by induction on $n$, that $w(n)\leq f(n):=\left((\gamma+1)^{d-2}\right)^{\log_2 n} \cdot n^2\cdot \log_2 n$. This is trivial when $n=1$, given that $w(1)=0$. Assume now that $n>1$. By Inequality~\ref{eq:substitution} and by induction, we get $w(n)\leq (\gamma+1)^{d-2} \cdot \sum\limits_{j,i}\left(\left((\gamma+1)^{d-2}\right)^{\log_2 n_{j,i}} \cdot n_{j,i}^2\cdot \log_2 n_{j,i}\right) + (\gamma+1)^{d-2} \cdot (n-1)\cdot \log_2 n$. Since $n_{j,i}\leq n/2<n$, we get $w(n)\leq (\gamma+1)^{d-2} \cdot \left((\gamma+1)^{d-2}\right)^{\log_2 (n/2)} \cdot \sum\limits_{j,i} n_{j,i}^2  \cdot \log_2 n + (\gamma+1)^{d-2} \cdot (n-1)\cdot \log_2 n=\left((\gamma+1)^{d-2}\right)^{\log_2 n} \cdot \sum\limits_{j,i} n_{j,i}^2  \cdot \log_2 n + (\gamma+1)^{d-2} \cdot (n-1)\cdot \log_2 n$. Since any two subtrees $T_{j,i}$ are disjoint, we get that $\sum n_{j,i} \leq n-1$, and hence $\sum\limits_{j,i} n_{j,i}^2\leq (n-1)^2$. Thus, $w(n)\leq \left((\gamma+1)^{d-2}\right)^{\log_2 n} \cdot ((n-1)^2+(n-1))\cdot \log_2 n \leq \left((\gamma+1)^{d-2}\right)^{\log_2 n} \cdot n^2\cdot  \log_2 n$. This completes the induction and the analysis of the width of $\Gamma$.
	
	{\bf Edge-length ratio.} By construction, the length of each edge connecting $r$ to a child is larger than or equal to $1$, hence the same is true for every edge of $T$. Thus, the edge-length ratio of $\Gamma$ is upper bounded by the maximum length of an edge of $T$. In turn, this is at most the sum of the height plus the width of $\Gamma$, which is in $\mathcal O\left(\left((\gamma+1)^{d-2}\right)^{\log_2 n} \cdot n^2\cdot  \log_2 n\right)$, as proved above. The factor $\left((\gamma+1)^{d-2}\right)^{\log_2 n}$ can be rewritten as $n^{(d-2)\cdot \log_2(\gamma+1)}$. The bound claimed in the statement is then obtained by substituting $\gamma=\lceil \frac{2}{\epsilon}\rceil$ and by observing that the value of $n$ used in the calculations is at most twice the size of the initial tree.	
\end{proof}

\section{Open Problems} \label{se:open-problems}


Our research raises a number of open problems which might be worth studying. 

First, the bounds in Theorem~\ref{th:tough} relating the toughness to the edge-length ratio of a drawing with constant spanning ratio are not tight; it would hence be \mbox{interesting to improve them.}

Second, we believe that there is still much to be understood about the edge-length ratio of planar straight-line drawings with constant spanning ratio. Theorem~\ref{th:trees-upperbound} shows that planar straight-line drawings with constant spanning ratio and polynomial edge-length ratio exist for bounded-degree trees. We also observe that every $n$-vertex $2$-connected outerplanar graph~$G$ admits a planar straight-line drawing with spanning ratio at most $\sqrt 2$ and edge-length ratio in $\mathcal O(n^{1.5})$; this can be achieved by placing the vertices of $G$, in the order given by the Hamiltonian cycle of $G$, at the vertices of a lattice $xy$-monotone polygonal curve; see, e.g.,~\cite{az-mne-95}. Further, it is known that Schnyder drawings are $2$-spanners~\cite{bfvv-cr-12}, hence every $n$-vertex \mbox{$3$-connected} planar graph admits a planar straight-line drawing with spanning ratio at most~$2$ and edge-length ratio in $\mathcal O(n)$; see~\cite{bfm-cd-07,s-epgg-90}. Do $3$-connected planar graphs (or even just maximal planar graphs) admit planar straight-line drawings with spanning ratio smaller than $2$ (and possibly arbitrarily close to $1$) and polynomial edge-length ratio? Is it possible to extend Theorem~\ref{th:tough} by proving that a planar straight-line drawing with constant spanning ratio and polynomial edge-length ratio exists \mbox{for every planar graph with bounded toughness?}

\bibliographystyle{splncs03} 
\bibliography{bibliography}

\end{document}